\documentclass{llncs}
\usepackage{graphicx}
\usepackage{amsmath}
\usepackage{amssymb}
\usepackage{amsmath}
\usepackage{amssymb}

\newcommand{\s}{\{0,1\}}
\newcommand{\tx}{\textsf}
\newtheorem{notat}{Notation}
\newtheorem{obs}{Observation}

\title{Balanced permutations Even-Mansour ciphers}
\author{Shoni Gilboa \inst {1}  \and Shay Gueron \inst {2, 3} \and Mridul Nandi \inst {4}}
\institute{
The Open University of Israel, Raanana 43107, Israel
\and
University of Haifa, Israel
\and
Intel Corporation, Israel Development Center, Israel
\and
Indian Statistical Institute, Kolkata
}

\begin{document}
\maketitle
\centerline{\today}

\begin{abstract}
The $r$-rounds Even-Mansour block cipher is a generalization of the well known Even-Mansour block cipher to $r$ iterations. 
Attacks on this construction were described by Nikoli{\'c} et al. and Dinur et al., for $r = 2, 3$.
These attacks are only marginally better than brute force, but are based on an interesting observation (due to Nikoli{\'c} et al.): for a ``typical'' permutation $P$, the distribution of  $P(x) \oplus x$ is not uniform.
This naturally raises the following question. Call permutations for which the distribution of $P(x) \oplus x$ is uniform ``balanced''. 
Is there a sufficiently large family of balanced permutations, and what is the security of the resulting Even-Mansour block cipher?

We show how to generate families of balanced permutations from the Luby-Rackoff construction, and use them to define a $2n$-bit block cipher from the $2$-rounds Even-Mansour scheme.
We prove that this cipher is indistinguishable from a random permutation of $\{0, 1\}^{2n}$, for any adversary who has oracle access to the public permutations and to an encryption/decryption oracle, as long as the number of queries is $o (2^{n/2})$. As a practical example, we discuss the properties and the performance of a $256$-bit block cipher that is based on our construction, and uses AES as the public permutation.
\end{abstract}

{\small
\begin{quote}
\end{quote}}

{\small
\begin{quote}
\end{quote}}

\section{Introduction}\label{intro}
The $r$-rounds Even-Mansour (EM) block cipher, suggested by Bogdanov et al. \cite{BKLSST}, encrypts an $n$-bit plaintext $m$ by 
\begin{equation}\label{eq:def_of_EM}
\textsf{EM}^{P_1,P_2,\ldots,P_r}_{K_0,K_1,\ldots,K_r}(m)=P_r(\ldots P_2(P_1(m\oplus K_0)\oplus K_1)\ldots\oplus K_{r-1})\oplus K_r,
\end{equation}
where $K_0,K_1,\ldots,K_r \in \{0, 1\}^n$ are secret keys and $P_1,P_2,\ldots,P_r$ are publicly known permutations, which are selected uniformly and independently at random, from the set of permutations of $\s^{n}$. 
The confidentiality of the EM cipher is achieved even though the permutations $P_1, \ldots, P_r$ are made public. 
For $r=1$, (\ref{eq:def_of_EM}) reduces to the classical Even-Mansour construction \cite{EM}.  

As a practical example, Bogdanov et al. defined the $128$-bit block cipher AES$^2$, which is an instantiation of the $2$-rounds EM cipher where the two public permutations are AES with two publicly known ``arbitrary'' keys (they chose the binary digits of the constant $\pi$). The complexity of the best (meet-in-the-middle) attack they showed uses $2^{129.6}$ cipher revaluations. Consequently, they conjectured that AES$^2$ offers $128$-bit security.

Understanding the security of the EM cipher has been the topic of extended research. 
First, Even and Mansour \cite{EM} proved, for $r=1$, that an adversary needs to make $\Omega(2^{n/2})$ oracle queries before he can decrypt a new message with high success probability. 
Daemen \cite{Daemen} showed that this bound is tight, by demonstrating a chosen-plaintext key-recovery attack after $O(2^{n/2})$ evaluations of $P_1$ and the encryption oracle.
Bogdanov et al. \cite{BKLSST}, showed, for the $r$-rounds EM cipher, $r\ge2$, that an adversary who sees only $O(2^{2n/3})$ chosen plaintext-ciphertext pairs cannot distinguish the encryption oracle from a random permutation of $\{0, 1\}^n$.
This result has been recently improved by Chen and Steinberger \cite{CS}, superseding intermediate progress made by Steinberger \cite{Steinberger} and by Lampe, Patarin and Seurin \cite{LPS}. They showed that for every $r$, an adversary needs $\Omega(2^{\frac{r}{r+1}n})$ chosen plaintext-ciphertext pairs before he can distinguish the $r$-rounds EM cipher from a random permutation of $\{0, 1\}^n$. This bound is tight, by Bogdanov et al.'s \cite{BKLSST} distinguishing attack after $O(2^{\frac{r}{r+1}n})$ queries.

Nikoli{\'c} et al. \cite{NWW} demonstrated a chosen-plaintext key-recovery attack on the single key variant ($K_0=K_1=K_2$) of the $2$-rounds EM cipher. Subsequently, Dinur et al. \cite{DDKS} produced additional key-recovery attacks on various other EM variants. 
All the attack in \cite{NWW} and \cite{DDKS} are only slightly better than a brute force approach. For example, the attack (\cite{DDKS}) on the single key variant of the $2$-rounds EM cipher has time complexity $O\left(\frac{\log n}{n}2^n\right)$, and the attack (\cite{DDKS}) on AES$^2$ (with three different keys) has complexity of $2^{126.8}$ 
(still better than Bogdanov et al. \cite{BKLSST}, thus enough to invalidate their that AES$^2$ has $2^{128}$ security). 

The above attacks are based on the astute observation, made in \cite{NWW}, that for a "typical" permutation $P$ of $\s^n$, the distribution of $P(x)\oplus x$ over uniformly chosen $x\in\s^n$ is not uniform. Currently, this observation yields only weak attacks, but the unveiled asymmetry may have the potential to lead to stronger results. 

This motivates the following question. 
Call a permutation $P$ of $\s^n$ ``balanced'' if the distribution of $P(x)\oplus x$, 
over uniformly chosen $x\in\s^n$, is uniform. Can we construct a block cipher based on balanced permutations?
We point out that, a priori, it is not even clear that there exists a family of such permutations, that is large enough to support a block cipher construction. 

In this work, we show how to generate a large family of balanced permutations of $\{0, 1 \}^{2n}$, by observing that a $2$-rounds Luby-Rackoff construction with any two identical \emph{permutations} of $\s^{n}$, always yields a balanced permutation (of $\s^{2n}$). We use these permutations in an EM setup 
(illustrated in Figure \ref{2rEM}, top panel), to construct a block cipher with block size of $2n$ bits. 
Note that in this EM setup, the permutations $P_1,P_2$ are not 
chosen uniformly at random from the set of all permutations of $\{0, 1 \}^{2n}$. They are selected from a particular subset of the permutations of $\{0, 1 \}^{2n}$, and defined via a random choice of two 
permutations of $\{0, 1 \}^{n}$, as the paper describes. 

For the security of the resulting $2n$ bits block cipher, we would ideally
like to maintain the security of the EM cipher (on blocks of $2n$ bits ). This would be guaranteed if we replaced the random permutation in the EM cipher, with an indifferentiable block cipher (as defined in \cite{MRH}). However, the balanced permutations we use in the EM construction are $2$-rounds Luby-Rackoff permutations, and it was shown in \cite{CPS} that even the $5$-rounds Luby-Rackoff construction does not satisfy indifferentiability. Therefore, it is reasonable to expect weaker security properties in our cipher. 
Indeed, we demonstrate a distinguishing  (not a key recovery) 
attack that uses $O (2^{n/2})$ queries. On the other hand, we prove that a smaller number of chosen plaintext-ciphertext queries is not enough to distinguish the block cipher from a random permutation of $\{0, 1\}^{2n}$. 

We comment that the combination of EM and Luby-Rackoff constructions have already been used and analyzed. 
Gentry and Ramzan \cite{GR} showed that the internal permutation of the Even-Mansour construction for $2n$-bits block size, can be securely replaced by a $4$-rounds Luby-Rackoff scheme with public round functions. They proved that the resulting construction is secure up to $O(2^{n/2})$ queries. 
Lampe and Seurin \cite{LS} discuss $r$-rounds Luby-Rackoff 
constructions where the round functions are of the form $x \mapsto F_i 
(K_i \oplus x)$, $F_i$ is a public random function, and $K_i$ is a (secret) round key. For an even number of rounds, this can be seen as a 
$r/2$-rounds EM construction, where the permutations are $2$-rounds 
Luby-Rackoff permutations. 
They show that this construction is secure up to $O(2^{\frac{tn}{t+1}})$ queries, where $t=\lfloor r/3\rfloor$ for non-adaptive chosen-plaintext adversaries, and $t=\lfloor r/6\rfloor$ for adaptive chosen-plaintext and ciphertext adversaries.
These works bare some similarities to ours, but the new feature in our construction is the emergence of balanced permutations. 

The paper is organized as follows. In Section \ref{sec:Balanced} we discuss 
balanced permutations and balanced permutations EM ciphers. Section \ref{sec:prelim} provides general background for the security analysis given in Section \ref{sec:proof}. In Section \ref{sec:attack}, we demonstrate the distinguishing attack. 
A practical use of our construction is a $256$-bit block cipher is based on AES. Section \ref{sec:practical} defines this cipher and discusses its performance characteristics. We conclude with a discussion in Section \ref{sec:discussion}.

\section{Balanced permutations and balanced permutations EM ciphers}
\label{sec:Balanced}

\subsection{Balanced permutations}

\begin{definition}[Balanced permutation\footnote{Also known as ``orthomorphism'' in the mathematical literature}] Let $\sigma$ be a permutation of $\s^n$. Define the function $\tilde{\sigma}:\s^n\rightarrow\s^n$ by $\tilde{\sigma}(\omega)=\omega\oplus\sigma(\omega)$, for every $\omega\in\s^n$. We say that $\sigma$ is a balanced permutation if $\tilde{\sigma}$ is also a permutation.
\end{definition}

\begin{example}
Let $A\in M_{n\times n}(\mathbb{Z}_2)$ be a matrix such that both $A$ and $I+A$ are invertible. Define $\pi_A:\mathbb{Z}_2^n\to\mathbb{Z}_2^n$ by $\pi_A(x)=Ax$. Then $\pi_A$ is a balance permutation of $\{ 0, 1 \}^n$.
One such matrix is defined by $A_{i,i} = A_{i, i+1} = 1$ for $i=1, 2, \ldots, n-1$, $A_{n, 1} = 1$ and $A_{i, j} = 0$ for all other $1 \le i, j \le n$.
\end{example}

\begin{example}
Let $a$ be an element of $GF(2^n)$ such that $a \ne 0, 1$. Identify $GF(2^n)$ with $\s^n$, so that field addition corresponds to bitwise XOR. The field's multiplication is denoted by $\times$. The function $x \to a \times x$ is a balanced permutation of $\{ 0, 1 \}^n$.
Note that this example is actually a special case of the previous one.
\end{example}
The balanced permutations provided by the above examples are a small family of permutations, and moreover are all linear. We now give a recipe for generating a large family of balanced permutations, by employing the Feistel construction that turns any function $f:\s^n\rightarrow\s^n$ to a permutation of $\s^{2n}$.

Let us use the following notation. For a string $\omega\in\s^{2n}$, denote the string of its first $n$ bits by $\omega_L\in\s^n$, and the string of its last $n$ bits by $\omega_R\in\s^n$. Denote the concatenation of two strings $\omega_1,\omega_2\in\s^n$ (in this order) by $\omega_1*\omega_2\in\s^{2n}$. We have the following identities:

\begin{equation*}
\left(\omega_1*\omega_2\right)_L=\omega_1,\quad\left(\omega_1*\omega_2\right)_R=\omega_2,\quad \omega_L*\omega_R=\omega.
\end{equation*}
\begin{definition}[Luby-Rackoff permutations]
Let $f:\s^n\rightarrow\s^n$ be a function. Let $\tx{LR}[f]:\s^{2n}\rightarrow\s^{2n}$ be the Luby-Rackoff (a.k.a Feistel) permutation
\begin{equation}\label{feistel}
\tx{LR}[f](\omega) :=\omega_R*\left(\omega_L\oplus f(\omega_R)\right).
\end{equation}
For every $r\geq 2$ and $r$ functions $f_1, \ldots, f_r: \s^n \rightarrow \s^n$, we define the $r$-rounds Luby-Rackoff permutation to be the composition
$$\tx{LR}[f_1, \ldots, f_r] :=\tx{LR}[f_r]\circ \cdots \circ\tx{LR}[f_1].$$
Since we use here extensively the special case $\tx{LR}[f, f]$, we denote it by $\tx{LR}^{\,2}[f]$.
\end{definition}

The following proposition shows that when $f$ is, itself, a permutation, then $\tx{LR}^{\,2}[f]$ is a balanced permutation.

\begin{proposition}
Let $f$ be a permutation of $\s^n$. Then, the $2$-rounds Luby-Rackoff permutation, $\tx{LR}^{\,2}[f]$, is a balanced permutation of $\s^{2n}$.
\label{2_Feistel_Rounds}
\end{proposition}

\begin{proof} Denote $P:=\tx{LR}^{\,2}[f]$. Observe first that
\begin{align}P(\omega)=\tx{LR}^{\,2}[f](\omega)&=\tx{LR}[f]\left(\tx{LR}[f](\omega)\right)=\tx{LR}[f]\left(\omega_R*\left(\omega_L\oplus f(\omega_R)\right)\right)=\nonumber\\
&=\left(\omega_L\oplus f(\omega_R)\right)*\left(\omega_R\oplus f\left(\omega_L\oplus f(\omega_R)\right)\right).\label{feistel2}
\end{align}
Therefore,
$$\tilde{P}(\omega)=f(\omega_R)* f\left(\omega_L\oplus f(\omega_R)\right).$$
Assume that $x,y\in\s^{2n}$ such that $\tilde{P}(x)=\tilde{P}(y)$, i.e.,
$$f(x_R)* f\left(x_L\oplus f(x_R)\right)=f(y_R)* f\left(y_L\oplus f(y_R)\right)$$
Then, $f(x_R)=f(y_R)$ and $f\left(x_L\oplus f(x_R)\right)=f\left(y_L\oplus f(y_R)\right)$.
Since (by assumption) $f$ is one-to-one, $x_R=y_R$ and $x_L\oplus f(x_R)=y_L\oplus f(y_R)$, it follows that
$x_L=\left(x_L\oplus f(x_R)\right)\oplus f(x_R)=\left(y_L\oplus f(y_R)\right)\oplus f(y_R)=y_L$.
We established that $\tilde{P}(x)=\tilde{P}(y)$ implies $x=x_L*x_R=y_L*y_R=y$ which concludes the proof.
\end{proof}
Figure \ref{2rFeistel} shows an illustration of $2$-rounds Luby-Rackoff (balanced) permutation.

\begin{figure}[ht!]
\centering
\includegraphics[width=100mm]{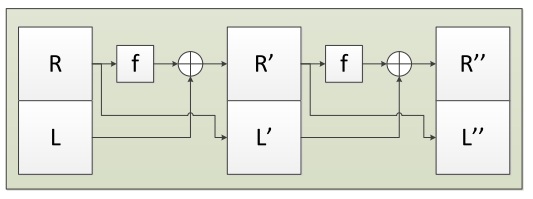}
\caption{
The figure shows a function from $\{0, 1\}^{2n}$ to $\{0, 1\}^{2n}$, based on two Feistel rounds with a function $f: \{0, 1\}^n \to \{0, 1\}^n$. For any function $f$, this construction is a permutation
of $\{0, 1\}^{2n}$, denoted $\tx{LR}^{\,2}[f]$. We call it a ``$2$-rounds Luby-Rackoff permutation''.
Proposition \ref{2_Feistel_Rounds} shows that if $f$ itself is a {\it permutation} of $\{0, 1\}^n$, then
$\tx{LR}^{\,2}[f]$ is a balanced permutation of $\{0, 1\}^{2n}$.}
\label{2rFeistel}
\end{figure}

\subsection{Balanced permutations EM ciphers}

\begin{definition}[$r$-rounds balanced permutations EM ciphers (BPEM)]
Let $n\ge 1$ and $r\ge 1$ be integers. Let $K_0, K_1, \ldots, K_r$ be $r+1$ strings in $\s^{2n}$. Let $f_1$,$f_2$,$\ldots$, $f_r$ be $r$ permutations of $\s^n$. Their associated $2$-rounds Luby-Rackoff (balanced) permutations (of $\s^{2n}$) are $\tx{LR}^{\,2}[f_1],\tx{LR}^{\,2}[f_2],\ldots,\tx{LR}^{\,2}[f_r]$, respectively. The $r$-rounds balanced permutations EM (BPEM) block cipher is  defined as
\begin{equation}\label{eq:def_of_BPEM}
\tx{BPEM}[K_0,K_1, \ldots, K_r;f_1, \ldots, f_r] := \tx{EM}^{\tx{LR}^{\,2}[f_1], \tx{LR}^{\,2}[f_2],\ldots,\tx{LR}^{\,2}[f_r]}_{K_0, K_1, \ldots, K_r},
\end{equation}
(where $ \tx{EM}$ is defined by \eqref{eq:def_of_EM}).
It encrypts $2n$-bit blocks with an $r$-rounds EM cipher with the keys $K_0, K_1, \ldots, K_r$, where the $r$ permutations $P_1,P_2,\ldots,P_r$ (of $\s^{2n}$) are set to be $\tx{LR}^{\,2}[f_1],$ $\tx{LR}^{\,2}[f_2],\ldots,\tx{LR}^{\,2}[f_r]$, respectively.
\end{definition}

The use of the $r$-rounds BPEM cipher for encryption (and decryption) starts with an initialization step, where the permutations $f_1,f_2,\ldots,f_r$ are selected uniformly and independently, at random from the set of permutations of $\s^{n}$. After they are selected, they can be made public. Subsequently, per session/message, the secret keys $K_0, K_1, \ldots, K_r$ are selected uniformly and independently, at random, from $\s^{2n}$. Figure \ref{2rEM} illustrates a $2$-rounds BPEM cipher 
$\tx{BPEM}[K_0, K_1, K_2;f_1, f_2]$, which is the focus of this paper.

\begin{remark}
\label{rem:linear}
The $r$-rounds EM cipher is not necessarily secure with {\it any} choice of balanced permutations as $P_1,P_2,\ldots,P_r$. For example, it can be easily broken when using the linear balanced permutations shown in Examples 1 and 2.
\end{remark}

\begin{remark}\label{rem:2LR}
In our construction, the permutations $P_1,P_2,\ldots,P_r$ are not random permutations. Therefore, the security analysis of the ``classical" EM does not apply, and the resulting cipher (BPEM) may not be secure. Indeed, it is easy to see that the $1$-round BPEM does not provide confidentiality. For any plaintexts $m\in\s^{2n}$, we have, by \eqref{feistel2},
$$\left(\tx{LR}^{\,2}[f](m\oplus K_0)\right)_L=\left(m_L\oplus (K_0)_L\right)\oplus f(m_R\oplus (K_0)_R)$$
Therefore, by \eqref{eq:def_of_BPEM}, \eqref{eq:def_of_EM} and \eqref{feistel2},
\begin{align*}\left(\tx{BPEM}[K_0,K_1;f](m)\right)_L&= \left(\tx{EM}^{\tx{LR}^{\,2}[f]}_{K_0,K_1}(m)\right)_L=\left(\tx{LR}^{\,2}[f](m\oplus K_0)\right)_L\oplus (K_1)_L = \\
&= m_L\oplus(K_0)_L\oplus(K_1)_L\oplus f(m_R\oplus (K_0)_R).
\end{align*}
It follows that if, e.g., $(m_1)_R=(m_2)_R$ then
$$\left(\tx{BPEM}[K_0,K_1;f](m_1)\oplus\tx{BPEM}[K_0,K_1;f](m_2)\right)_L=(m_1\oplus m_2)_L$$ which means that the ciphertexts leak out information on $m_1,m_2$.
This also implies that the $r$-rounds BPEM cipher must be used with $r \ge2$ to have any hope for achieving security.
\end{remark}

\begin{remark}\label{rem:immune}
By construction, $\tx{BPEM}[K_0,K_1, \ldots, K_r;f_1, \ldots, f_r]$ ($r \ge 2$) is immune
against any attack that tries to leverage the non-uniformity of $P(x) \oplus x$ (including \cite{NWW} and \cite{DDKS})). Obviously, this does not guarantee it is secure (as indicated in Remark \ref{rem:linear}).
\end{remark}

In Section \ref{sec:proof} we prove that the $2$-round BPEM cipher is indistinguishable from a random permutation.

\begin{figure}[ht!]
\centering
\includegraphics[width=0.50\textwidth]{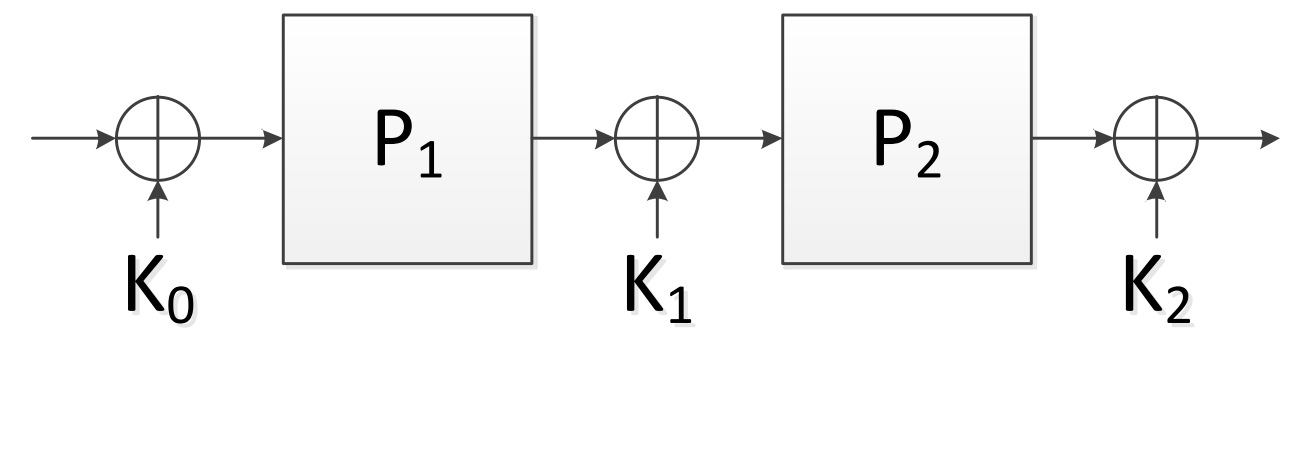}
\includegraphics[width=0.99\textwidth]{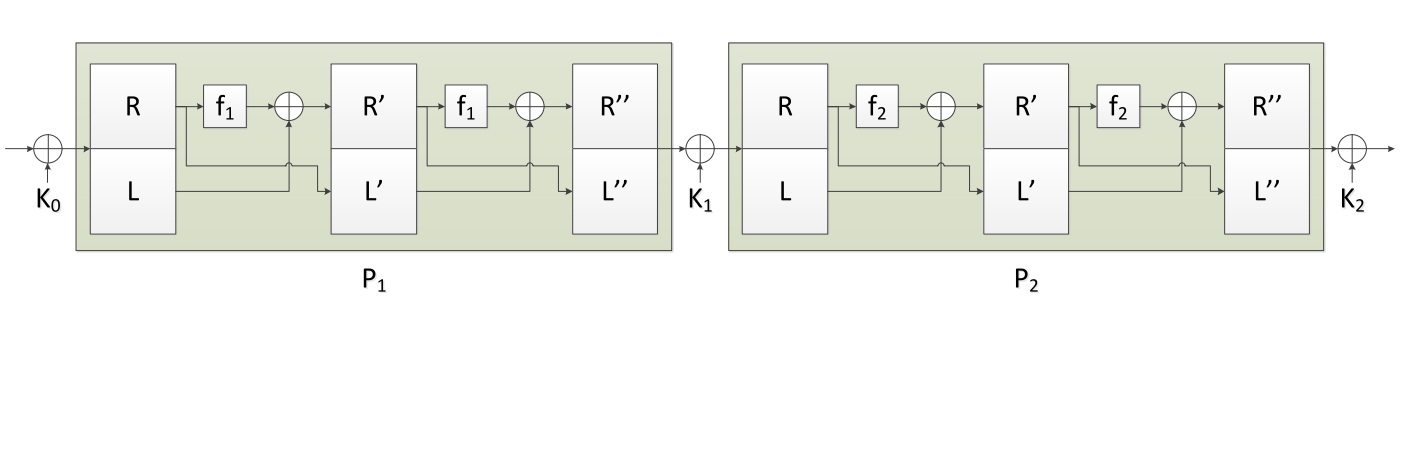}
\vspace{-1.5cm}
\caption{The $2$-rounds balanced permutations EM (BPEM) cipher operates on blocks of size $2n$ bits.
The permutations $P_1$ and $P_2$ are balanced permutations of $\{0, 1\}^{2n}$, defined as $2$-rounds Luby-Rackoff permutations. $f_1$ and $f_2$ are two (public) permutations of $\{ 0, 1 \}^{n}$.
Each of $K_0,K_1,K_2$ is a $2n$-bit secret key. See explanation in the text.}
\label{2rEM}
\end{figure}

\subsection{Equivalent representation of BPEM in terms of LR}\label{sec:BPEM_LR}
In this section we show that 2-rounds BPEM can be viewed as a ``keyed"\footnote{By ``keyed" we mean that each function used in the Luby-Rackoff construction is selected from a family of functions indexed by a key.} Luby-Rackoff cipher (in fact, the $r$-rounds BPEM has a similar representation for every $r$).
\begin{notat}\label{notat:f_oplus_K}
Given a function $f: \s^n \to \s^n$ and a key $K\in\s^n$ we denote $\tx{EM}^{f}_{K,K}$ by $f^{\oplus K}$, namely
$$\tx{EM}^{f}_{K,K}(x)=f(x\oplus K)\oplus K.$$
\end{notat}
\begin{lemma}\label{lem:BPEM=LR}
Let $K_0, K_1, K_2 \in \s^{2n}$ and let $f_1, f_2$ be two permutations of $\{0,1 \}^n$. Then,
\begin{equation*}\label{eq:BPEM=LR}
\tx{BPEM}[K_0, K_1, K_2;f_1,f_2] = \tx{LR}[f_1^{\oplus K'_1}, f_1^{\oplus K'_2}, f_2^{\oplus K'_3}, f_2^{\oplus K'_4}] \oplus (K'_6 * K'_5)
\end{equation*}
where
\begin{equation}\label{matrix1}
\begin{pmatrix}
K'_{1} \\
K'_{2} \\
K'_{3} \\
K'_{4} \\
K'_{5} \\
K'_{6}
\end{pmatrix}
= \begin{pmatrix}
1 & 0 & 0 & 0 & 0 & 0\\
1 & 1 & 0 & 0 & 0 & 0 \\
0 & 1 & 1 & 0 & 0 & 0 \\
1 & 0 & 1 & 1 & 0 & 0 \\
1 & 1 & 0 & 1 & 1 & 0 \\
1 & 0 & 1 & 1 & 0 & 1
\end{pmatrix}
\cdot
\begin{pmatrix}
(K_0)_R \\
(K_0)_L \\
(K_1)_R \\
(K_1)_L\\
(K_2)_R \\
(K_2)_L
\end{pmatrix}
.
\end{equation}
\end{lemma}
\begin{proof}
 For every function $f:\s^n\to\s^n$, $K\in\s^{2n}$ and $\omega\in\s^{2n}$ we have, by \eqref{feistel},
\begin{align*}
\tx{LR}[f](\omega\oplus K)&=(\omega\oplus K)_R*\left((\omega\oplus K)_L\oplus f((\omega\oplus K)_R)\right)=\\
&=\left(\omega_R*\left(\omega_L\oplus f(\omega_R\oplus K_R)\oplus K_R\right)\right)\oplus\left(K_R*(K_L\oplus K_R)\right)=\\
&=\left(\omega_R*\left(\omega_L\oplus f^{\oplus K_R}(\omega_R)\right)\right)\oplus\left(K_R*(K_L\oplus K_R)\right)=\\
&=\tx{LR}\left[f^{\oplus K_R}\right](\omega)\oplus\left(K_R*(K_L\oplus K_R)\right)
\end{align*}
and hence
\begin{align*}
\tx{LR}^{\,2}[f](\omega\oplus K)&=\tx{LR}[f]\left(\tx{LR}[f](\omega\oplus K)\right)=\\
&=\tx{LR}[f]\left(\left(\tx{LR}\left[f^{\oplus K_R}\right](\omega)\right)\oplus\left(K_R*(K_L\oplus K_R)\right)\right)=\\
&=\tx{LR}\left[f^{\oplus (K_L\oplus K_R)}\right]\left(\tx{LR}\left[f^{\oplus K_R}\right](\omega)\right)\oplus\left((K_L\oplus K_R)*K_L\right)=\\
&=\tx{LR}\left[f^{\oplus K_R},f^{\oplus (K_L\oplus K_R)}\right](\omega)\oplus\left((K_L\oplus K_R)*K_L\right).
\end{align*}
In particular
\begin{align*}
\tx{LR}^{\,2}[f_1]&(\omega \oplus K_0)=\\
&=\tx{LR}\left[f_1^{\oplus (K_0)_R},f_1^{\oplus \left((K_0)_L\oplus (K_0)_R\right)}\right](\omega)\oplus\left(((K_0)_L\oplus (K_0)_R)*(K_0)_L\right)=\\
&=\tx{LR}\left[f_1^{\oplus K'_1},f_1^{\oplus K'_2}\right](\omega)\oplus\left(K'_2*(K'_1\oplus K'_2)\right)
\end{align*}
and then
\begin{align*}
\tx{LR}^{\,2}&[f_2]\left(\tx{LR}^{\,2}[f_1](\omega \oplus K_0)\oplus K_1\right)=\\
=&\tx{LR}^{\,2}[f_2]\left(\tx{LR}\left[f_1^{\oplus K'_1},f_1^{\oplus K'_2}\right](\omega)\oplus\left(K'_2*(K'_1\oplus K'_2)\right)\oplus K_1\right)=\\
=&\tx{LR}\left[f_2^{\oplus \left(K'_1\oplus K'_2\oplus(K_1)_R\right)},f_2^{\oplus \left(K'_1\oplus(K_1)_L\oplus(K_1)_R\right)}\right]\left(\tx{LR}\left[f_1^{\oplus K'_1},f_1^{\oplus K'_2}\right](\omega)\right)\oplus\\
&\oplus\left((K'_1\oplus(K_1)_L\oplus(K_1)_R)*(K'_2\oplus(K_1)_L)\right)=\\
=&\tx{LR}\left[f_1^{\oplus K'_1},f_1^{\oplus K'_2},f_2^{\oplus K'_3},f_2^{\oplus K'_4}\right](\omega)\oplus\left(K'_4*(K'_3\oplus K'_4)\right).
\end{align*}
Therefore, by \eqref{eq:def_of_BPEM} and \eqref{eq:def_of_EM},
\begin{align*}
\tx{BPEM}&\left[K_0, K_1, K_2;f_1,f_2\right](\omega)=\tx{EM}^{\tx{LR}^{\,2}[f_1],\tx{LR}^{\,2}[f_2]}_{K_0, K_1, K_2}(\omega)=\\
=&\tx{LR}^{\,2}[f_2]\left(\tx{LR}^{\,2}[f_1](\omega\oplus K_0)\oplus K_1\right)\oplus K_2=\\
=&\tx{LR}\left[f_1^{\oplus K'_1},f_1^{\oplus K'_2},f_2^{\oplus K'_3},f_2^{\oplus K'_4}\right](\omega)\oplus\left((K'_4\oplus(K_2)_L)*(K'_3\oplus K'_4\oplus(K_2)_R)\right)=\\
=& \tx{LR}\left[f_1^{\oplus K'_1}, f_1^{\oplus K'_2}, f_2^{\oplus K'_3}, f_2^{\oplus K'_4}\right](\omega) \oplus (K'_6 * K'_5).
\end{align*}
\end{proof}

\section{Security preliminaries and definitions}
\label{sec:prelim}
Let $A$ be an oracle adversary which interacts with one or more oracles.
Suppose that $\mathcal{O}$ and $\mathcal{O}'$ are two oracles (or a tuple of oracles) with same domain and range spaces. We define the distinguishing advantage of $A$ distinguishing $\mathcal{O}$ and $\mathcal{O}'$ as
$$\Delta_{A}(\mathcal{O} ;\mathcal{O}') := \big|\Pr[A^{\mathcal{O}} = 1] - \Pr[A^{\mathcal{O'}} = 1]\big|.$$
The maximum advantage $\max_{A} \Delta_{A}(\mathcal{O} ; \mathcal{O}')$ over all adversaries with complexity $\theta$ (which includes query, time complexities etc.) is denoted by $\Delta_{\theta}(\mathcal{O};\mathcal{O}')$. When we consider computationally unbounded adversaries (which is done in this paper), the time and memory parameters are not present and so we only consider query complexities.
In the case of a single oracle, $\theta$ is the number of queries, and in the case of a tuple of oracles, $\theta$ would be of the form $(q_1, \ldots, q_r)$ where $q_i$ denotes the number of queries to the $i^{\rm{th}}$ oracle.
While we define security advantages of $\mathcal{O}$, we usually choose $\mathcal{O}'$ to be an ideal candidate, such as the random permutation $\Pi$ or a random function.
The PRP-advantage of $A$ against a keyed construction $\mathcal{C}_K$ is $\Delta_{A}(\mathcal{C}_K ; \Pi)$. The maximum PRP-advantage with query complexity $\theta$ is denoted as $\Delta_{\mathcal{C}}^{\mathrm{prp}}(\theta)$.

In this paper, we always assume that queries to an oracle $\mathcal{O}$ are allowed in both directions, i.e., to $\mathcal{O}^{-1}$ as well. We denote
\begin{align*}
\Delta^{\pm}_A(\mathcal{O},\mathcal{O'})&:=\Delta_A\left((\mathcal{O},\mathcal{O}^{-1});(\mathcal{O'},\mathcal{O'}^{-1})\right),\\
\Delta^{\pm}_{\theta}(\mathcal{O},\mathcal{O'})&:=\Delta_{\theta}\left((\mathcal{O},\mathcal{O}^{-1});(\mathcal{O'},\mathcal{O'}^{-1})\right).
\end{align*}
The SPRP-advantage of a keyed construction $\mathcal{C}_K$ (where the adversary has access to both the encryption $\mathcal{C}_K$ and its decryption $\mathcal{C}_K^{-1}$) is defined by $$\Delta_{\mathcal{C}}^{\mathrm{sprp}}(\theta): = \Delta^{\pm}_{\theta}(\mathcal{C}_K; \Pi).$$
When a construction $\mathcal{C}$ is based on one or more ideal permutations or random permutations $f_1, \ldots, f_r$ and a key $K$, we define SPRP-advantage of a distinguisher $A$, in presence of ideal candidates,
as $\Delta^{\pm}_{A}((\mathcal{C}, f_1,\ldots, f_r); (\Pi, f_1, \ldots, f_r))$ where $\Pi$ is sampled independently of $\hat{f} := (f_1, \ldots, f_r)$. We denote the maximum advantage by $\Delta_{\mathcal{C}}^{\mathrm{im\textit{-}sprp}}(\theta) := \Delta^{\pm}_{\theta}((\mathcal{C}, \hat{f}); (\Pi, \hat{f}))$ which we call SPRP-advantage in the ideal model.
The complexity parameters of the above advantages depend on the number of oracles, and will be explicitly declared in specific instances.

We state two simple observations on the distinguishing advantages for oracles (we skip the proofs of these observations, as these are straightforward).
\begin{obs}\label{obs1}
If $\mathcal{O}_1$, $\mathcal{O}_2$ and $\mathcal{O}'$ are three independent oracles, then
$$\Delta^{\pm}_{q,q'}\left((\mathcal{O}_1, \mathcal{O}');(\mathcal{O}_2, \mathcal{O}')\right) \leq \Delta^{\pm}_{q}(\mathcal{O}_1;\mathcal{O}_2).$$
\end{obs}
\begin{obs}\label{obs2}
If $\mathcal{C}$ is an oracle construction, then (by using standard reduction)
$$\Delta^{\pm}_{q,q'}\left((\mathcal{C}^{\mathcal{O}_1}, \mathcal{O}');(\mathcal{C}^{\mathcal{O}_2}, \mathcal{O}')\right) \leq \Delta^{\pm}_{rq,q'}\left((\mathcal{O}_1, \mathcal{O}');(\mathcal{O}_2, \mathcal{O}')\right)$$
(where $r$ is the number of queries to $\mathcal{O}$, needed to simulate one query to the construction $\mathcal{C}^{\mathcal{O}}$).
\end{obs}

Note that in the Observation \ref{obs2}, we do not need to assume any kind of independence of the oracles. Analogous observations, up to obvious changes, hold for the case where $\mathcal{O}_1,\mathcal{O}_2,\mathcal{O}'$ are tuples of oracles.

\subsection{Coefficient-H technique}
Patarin's coefficient-H technique \cite{Patarin1} (see also \cite{Patarin3}) is a tool for showing an upper bound for the distinguishing advantage. Here is the basic result of the technique.

\begin{theorem}[Patarin \cite{Patarin1}]
\label{fact:H}
Let $\mathcal{O}$ and $\mathcal{O}'$ be two oracle algorithms with domain $D$ and range $R$. Suppose there exist a set $\mathcal{V}_{bad} \subseteq D^q \times R^q$ and $\varepsilon > 0$ such that the following conditions hold:
\begin{enumerate}
\item
For all $(x_1, \ldots, x_q, y_1$, $\ldots$, $y_q) \not\in \mathcal{V}_{bad}$,
\[\Pr[\mathcal{O}(x_1) = y_1, \ldots, \mathcal{O}(x_q) = y_q] \geq (1 - \varepsilon) \Pr[\mathcal{O}'(x_1) = y_1, \ldots, \mathcal{O}'(x_q) = y_q]\]
(the above probabilities are called interpolation probabilities).
\item
For all $A$ making at most $q$ queries to $\mathcal{O}'$, $\Pr[\mathrm{Trans}(A^{\mathcal{O}'}) \in \mathcal{V}_{bad}] \leq \delta$ where $\mathrm{Trans}(A^{\mathcal{O}'}) = (x_1, \ldots, x_q, y_1, \ldots, y_q)$, $x_i$ and $y_i$ denote the $i^{\rm{th}}$  query and response of $A$ to $\mathcal{O}'$.
\end{enumerate}
Then,
\[\Delta_q(\mathcal{O};\mathcal{O}') \leq \varepsilon + \delta.\]
\end{theorem}

The above result can be applied for more than one oracle. In such cases we split the parameter $q$ into $(q_1, \ldots, q_r)$ where $q_i$ denotes the maximum number of queries to the $i^{\rm{th}}$ oracle. Moreover, if we have an oracle $\mathcal{O}$ and its inverse $\mathcal{O}^{-1}$ then the interpolation probability for both $\mathcal{O}$ and $\mathcal{O}^{-1}$ can be simply expressed through the interpolation probability of $\mathcal{O}$ only. For example, if an adversary makes a query $y$ to $\mathcal{O}^{-1}$ and obtains the response $x$, we can write $\mathcal{O}(x) = y$.
Therefore, under the conditions of Theorem \ref{fact:H} we also have $\Delta^{\pm}_q(\mathcal{O};\mathcal{O}') \leq \varepsilon + \delta$.

\subsection{Known related results}\label{subsection3_2}

\subsubsection{The security of Even-Mansour cipher} 
It is known that the Even-Mansour cipher $\tx{EM}_{K_0, K_1}$ is SPRP secure in the ideal model, in the following sense:  $\Delta_{EM}^{\mathrm{im\textit{-}sprp}}(q_1, q_2) = O(q_1 q_2/2^n)$. The same is true for the single key variant $\tx{EM}_{K, K}$.
In Section \ref{sec:proof}, we provide (using Patarin's coefficient-H technique) a simple proof of this result (Lemma \ref{lemma2}) and a more general result (Lemma \ref{lemma3}).

\subsubsection{The security of Luby-Rackoff encryption}
The $3$-rounds Luby-Rackoff construction is PRP secure and the $4$-rounds Luby-Rackoff construction is SPRP secure ,
when the underlying functions $f_i$ are PRP's (or PRF's).
We use the following quantified version of the SPRP security of the $4$-rounds case.
\begin{theorem}[Piret \cite{Piret}]
\label{Piret}
Let $\Pi_1, \ldots, \Pi_4$ be four independent random permutations of
$\s^n$, and let $\Pi$ be a random permutation of $\s^{2n}$.
Then, $\tx{LR}[\Pi_1, \ldots, \Pi_4]$ is SPRP  secure in the following sense:
\[\Delta^{\pm}_{q}(\tx{LR}[\Pi_1, \ldots, \Pi_4];\Pi) \leq \frac{5q (q-1)}{2^n}.\]
\end{theorem}

The above bound $O(q^2/2^n)$ is tight (see \cite{TP}).
In the proof of Theorem \ref{thm:indistinguish_1key_1perm} we use the following, more general, result.
\begin{theorem}[Nandi \cite{Nandi}]
\label{Nandi}
Let $ r \geq 4$, and let $(\alpha_1, \ldots \alpha_r)$  be a sequence of numbers from $\{1,\ldots, t\}$ such that $(\alpha_1, \ldots \alpha_r) \neq (\alpha_{r}, \ldots, \alpha_1)$. Let $\Pi_1, \ldots, \Pi_t$ be $t$ independent random permutations of $\s^n$, and let $\Pi$ be a random permutation of $\s^{2n}$. Then, $\tx{LR}[\Pi_{\alpha_1}, \ldots, \Pi_{\alpha_r}]$ is  SPRP  secure in the following sense:
\[\Delta_{q}(\tx{LR}[\Pi_{\alpha_1}, \ldots, \Pi_{\alpha_r}];\Pi) \leq \frac{(r^2+1)q^2}{2^n-1} + \frac{q^2}{2^{2n}}.\]
\end{theorem}

\section{Security analysis of our construction}\label{sec:proof}

\subsection{Security analysis of tuples of single key $1$-round  EM cipher}

\begin{notat}
Let $x_1, \ldots, x_t \in \s^n$.
We use $\tx{coll}(x_1, \ldots, x_t)$ to indicate the existence of a collision, i.e., that
$x_i = x_j$ for some $1 \le i < j \le t$. Otherwise, we write $\tx{dist}_n(x_1, \ldots, x_t)$, and say that the tuple $(x_1, \ldots, x_t)$ is block-wise distinct. Given a function $f : \s^n \rightarrow \s^n$ and a tuple $x_1, \ldots, x_t \in \s^n$ we define $$f^{(t)}(x_1, \ldots, x_t) := (f(x_1), \ldots, f(x_t)).$$
For positive integers $m,r$, denote
$$P(m,r) = m(m-1) \cdots (m-r+1).$$
\end{notat}

\begin{obs}
For every $\tx{dist}_n(x_1, \ldots, x_t)$, $\tx{dist}_n(y_1, \ldots, y_t)$ and a uniform random permutation $\Pi$ on $\s^n$,
\begin{equation*}\label{eq:rp}
\Pr[\Pi^{(t)}(x_1, \ldots, x_t) = (y_1, \ldots, y_t)] = \frac{1}{P(2^n, t)}
\end{equation*}
More generally, let $\Pi_1, \ldots, \Pi_r$ be independent uniform random permutations over $\s^n$ then for every block-wise distinct tuples $X^i, Y^i \in (\s^n)^{t_i}$, $1 \leq i \leq r$ we have
\begin{equation}\label{eq:rp-general}
\Pr[\Pi_1^{(t_1)}(X^1) = Y^1, \ldots, \Pi_r^{(t_r)}(X^r) = Y^r] = \frac{1}{P(2^n,t_1)} \times \cdots \times \frac{1}{P(2^n, t_r)}.
\end{equation}
\end{obs}

Now we show that for a random permutation $\Pi$ of $\s^n$ and a uniformly chosen $K$, the permutation $\Pi^{\oplus K}$ (single keyed $1$-round EM, see Notation \ref{notat:f_oplus_K}) is SPRP secure in the ideal model.
\begin{lemma}
\label{lemma2}
Let $\Pi$ and $\Pi_1$ be independent random permutations of $\s^n$. Then
\[\Delta^{\pm}_{q_1, q_2}\left((\Pi^{\oplus K}, \Pi);(\Pi_1, \Pi)\right) \leq  \frac{2q_1q_2}{2^n}.\]
\end{lemma}
\begin{proof} We use Patarin's coefficient H-technique. We take the set of bad views $\mathcal{V}_{bad}$ to be the empty set.
We need to show that for every tuples $M, C \in (\s^n)^{q_1}$, $x, y \in (\s^n)^{q_2}$,
\begin{gather*}
\Pr[\Pi^{\oplus K}(M_i) = C_i, 1 \leq i \leq q_1,\, \Pi(x_i) = y_i, 1 \leq j \leq q_2]\geq\\
\geq (1 - \varepsilon) \Pr[\Pi_1(M_i) = C_i, 1 \leq i \leq q_1,\, \Pi(x_i) = y_i, 1 \leq j \leq q_2],
\end{gather*}
where $\varepsilon=\frac{2q_1q_2}{2^n}$.
With no loss of generality we may assume that each of the tuples $M,C,x,y$ is block-wise distinct. Then, by \eqref{eq:rp-general},
\begin{gather*}I_{ideal} := \Pr[\Pi_1(M_i) = C_i, 1 \leq i \leq q_1,\, \Pi(x_i) = y_i, 1 \leq j \leq q_2]=\\
= \Pr[\Pi_1^{(q_1)}(M) = C, \Pi^{(q_2)}(x) = y] = \frac{1}{P(2^n, q_1)} \times \frac{1}{P(2^n, q_2)}.
\end{gather*}
We say that a key $K \in \s^n$ is ``good'' if $K \oplus M_i \neq x_j$  and $K \oplus C_i \neq y_j$ for all $1 \le i \le q_1$, $1 \le j \le q_2$. In other words, for a good key all the inputs (outputs) of $\Pi$ (in the $I_{real}$ computation) are block-wise distinct. Thus, for any given good key $K$,
\begin{gather*}
 \Pr[\Pi(M_i \oplus K) = (K \oplus C_i), 1 \leq i \leq q_1, \Pi(x_j) = y_j, 1 \leq j \leq q_2]=\\
= \frac{1}{P(2^n, q_1 + q_2)} \geq I_{ideal}.
\end{gather*}
By a simple counting argument, the number of good keys is at least $2^n - 2q_1q_2$, i.e., the probability that a randomly chosen key is good, is at least $(1 - \varepsilon)$, where $\varepsilon = \frac{2q_1q_2}{2^n}$. Therefore, we have
\begin{gather*}
I_{real}:= \Pr[\Pi^{\oplus K}(M_i) = C_i, 1 \leq i \leq q_1,\, \Pi(x_i) = y_i, 1 \leq j \leq q_2]\geq (1 - \varepsilon) I_{ideal}
\end{gather*}
and the result follows by Theorem \ref{fact:H}.
\end{proof}

Now, we extend Lemma \ref{lemma2} to a tuple $(\Pi_{\alpha_1}^{\oplus K_{\beta_1}}, \ldots, \Pi_{\alpha_t}^{\oplus K_{\beta_t}})$ of single key $1$-round EM encryptions, where some keys and permutations can be repeated.
\begin{lemma}\label{lemma3}
Let $\Pi_1, \ldots, \Pi_r, \bar{\Pi}_1, \ldots, \bar{\Pi}_t$ be independent random permutations of $\s^n$ and $K_1, \ldots K_s$ be chosen uniformly and independently from $\s^n$. We write $\hat{\Pi}$ to denote $(\Pi_1, \ldots, \Pi_r)$. Let $(\alpha_1, \ldots, \alpha_t)$ and  $(\beta_1, \ldots, \beta_t)$ be a sequence of elements from $\{1,\ldots,r\}$ and $\{1,\ldots,s\}$, respectively, such that $(\alpha_i, \beta_i)$'s are distinct. Then, for any $\theta = (q_1, \ldots, q_t, q'_1, \ldots, q'_r)$ (specifying the maximum number of queries for each permutation) we have
\[\Delta^{\pm}_{\theta}\left((\bar{\Pi}_1, \ldots, \bar{\Pi}_t, \hat{\Pi});(\Pi_{\alpha_1}^{\oplus K_{\beta_1}}, \ldots, \Pi_{\alpha_t}^{\oplus K_{\beta_t}}, \hat{\Pi})\right) \leq  \frac{\sigma}{2^n}\]
where
$\sigma :=2\sum_{\alpha=1}^r\left(\binom{\sigma_{\alpha}}{2}+\sigma_{\alpha}q'_{\alpha}\right)$
and $\sigma_{\alpha}=\sum_{\substack{1\leq i\leq t\\\alpha_i=\alpha}}q_i$ for every $1\leq\alpha\leq r$.
\end{lemma}
\begin{proof}
The proof is similar to the proof of Lemma \ref{lemma2}. Let $M^i, C^i \in (\s^n)^{q_i}$, $1\leq i\leq t$, $X^{\alpha}, Y^{\alpha} \in (\s^n)^{q'_{\alpha}}$, $1 \leq \alpha \leq r$, be block-wise distinct tuples. From \eqref{eq:rp-general}, we have that
\begin{gather*}
I_{ideal}  = \Pr[\bar{\Pi_i}^{(q_i)}(M^i) = C^i, 1 \leq i \leq t,\, \Pi_{\alpha}^{(q'_{\alpha})}(X^{\alpha}) = Y^{\alpha}, 1 \leq\alpha \leq r]=\\
=\prod_{i=1}^t \frac{1}{P(2^n, q_i)} \times \prod_{\alpha=1}^r \frac{1}{P(2^n, q'_{\alpha})}.
\end{gather*}
We say that a tuple of keys $(K_1, \ldots, K_s)$ is ``bad'' if one of the following holds:
\begin{enumerate}
 \item
There are $1 \le i, i' \le t$, $1 \le j \le q_i$,
$1 \le j' \le q_{i'}$ such that
$(i, j) \neq (i', j')$,
$\alpha_i = \alpha_{i'}$, and
$K_{\beta_i} \oplus M^{\alpha_i}_j = K_{\beta_{i'}} \oplus M^{\alpha_{i'}}_{j'}$.
 \item
There are $1 \le i \le t$, $1 \le j \le q_i$,
$1 \le j' \le q'_{\alpha_i}$
such that
$K_{\beta_i} \oplus M^{\alpha_i}_j = X^{\alpha_i}_{j'}$.
 \item
There are $1 \le i, i' \le t$, $1 \le j \le q_i$,
$1 \le j' \le q_{i'}$ such that
$(i, j) \neq (i', j')$, $\alpha_i = \alpha_{i'}$, and
$K_{\beta_i} \oplus C^{\alpha_i}_j = K_{\beta_{i'}} \oplus C^{\alpha_{i'}}_{j'}$.
 \item
There are $1 \le i \le t$, $1 \le j \le q_i$,
$1 \le j' \le q'_{\alpha_i}$
such that
$K_{\beta_i} \oplus C^{\alpha_i}_j = Y^{\alpha_i}_{j'}$.
\end{enumerate}
Note that there are at most $\sum_{\alpha=1}^r\binom{\sigma_{\alpha}}{2}$
cases in the first and in the third items, and at most
$ \sum_{\alpha=1}^r \sigma_{\alpha}q'_{\alpha}$ cases in the second and fourth items.

If a key tuple is not bad, we say that it is a ``good'' key tuple.
As in the proof of Lemma \ref{lemma2},
for a good key tuple all the inputs (outputs) of each permutation are distinct.
Thus, given a good tuple of keys $(K_1, \ldots, K_s)$, it is easy to see that
\begin{gather*}
\Pr[(\Pi_{\alpha_i}^{\oplus K_{\beta_i}})^{(q_i)}(M^i) = C^i, 1 \leq i \leq t, \Pi_{\alpha}^{(q'_{\alpha})}(X^{\alpha}) = Y^{\alpha}, 1 \leq \alpha\leq r]=\\
= \prod_{\alpha=1}^r \frac{1}{P(2^n, \sigma_{\alpha}+q'_{\alpha})} \geq I_{ideal}.
\end{gather*}
It now remains to bound the probability that a random key tuple is bad. This can happen with one of the cases listed in items 1-4 where each case has probability $2^{-n}$ to occur. Hence, the probability that a random key tuple is bad, is at most
$\frac{\sigma}{2^{n}}$, and the probability that a random key tuple is good is therefore at least $1 - \frac{\sigma}{2^{n}}$. The result follows by Theorem \ref{fact:H}.
\end{proof}

\subsection{Main theorems}
\begin{theorem}\label{thm:indistinguish}
Consider the BPEM cipher $\tx{BPEM}[K_0, K_1, K_2;f_1,f_2]$ where the (secret) keys
$K_0, K_1, K_2$ are selected uniformly and independently at random.
Let $q_{*}$ be the maximum number of queries to the encryption/decryption oracle, and let $q_1, q_2$ be the maximum numbers of queries to the public permutations $f_1$ and $f_2$, respectively. Then,
$$
\Delta^{im\textit{-}sprp}_{\tx{BPEM}}(q_{*},q_1, q_2)\leq
\frac{q_{*}(13q_{*}+4q_1+4q_2)}{2^n}.
$$
\end{theorem}

\begin{proof}
By Lemma \ref{lem:BPEM=LR}, we know that our BPEM construction is same as $$\tx{LR}[f_1^{\oplus K'_1}, f_1^{\oplus K'_2}, f_2^{\oplus K'_3}, f_2^{\oplus K'_4}] \oplus (K'_6 * K'_5),$$
where $K'_1, \ldots, K'_6$ are defined via \eqref{matrix1} by $K_1,K_2,K_3,K_4$. The matrix in \eqref{matrix1} is lower triangular with non-zero diagonal, and hence non-singular. Thus, the ``new" keys $K'_1, \ldots, K'_6$ are also distributed uniformly and independently. As $K'_5, K'_6$ are independent of all the ``ingredients" of  $\tx{LR}[f_1^{\oplus K'_1}, f_1^{\oplus K'_2}, f_2^{\oplus K'_3}, f_2^{\oplus K'_4}]$, it suffices to prove our result without the keys $K'_5$ and $K'_6$.

Let $\Pi_1, \ldots, \Pi_4$ be random permutations of $\s^n$ and let $\Pi$ be a random permutation of $\s^{2n}$, all are independent of each other and independent of
$\hat{f} = (f_1, f_2)$).
Denote $\hat{F} = (f_1^{\oplus K'_1},f_1^{\oplus K'_2},f_2^{\oplus K'_3}, f_2^{\oplus K'_4})$  and  $\hat{\Pi} = (\Pi_1, \ldots, \Pi_4)$.
By Observation \ref{obs2} and Lemma \ref{lemma3}, we have\footnote{
Note that each query to the oracle construction $\tx{LR}[g_1,g_2,g_3,g_4]$ translates to four queries - one to each permutation
 $g_i$, $i=1, \ldots, 4$}
\begin{multline*}
\Delta^{\pm}_{q_{*},q_1,q_2}\left((\tx{LR}[\hat{F}], \hat{f});(\tx{LR}[\hat{\Pi}], \hat{f})\right) \leq\\
\leq\Delta^{\pm}_{q_{*},q_{*},q_{*},q_{*},q_1,q_2}\left((\hat{F}, \hat{f});(\hat{\Pi}, \hat{f})\right)
\leq  \frac{4 q_{F} \left(2q_{F} +q_1+q_2\right)}{2^n}.
\end{multline*}
Finally, by applying the triangle inequality,  Observation \ref{obs1} and Theorem \ref{Piret}, the SPRP-advantage in the ideal model is
\begin{align*}
\Delta^{\pm}&_{q_{*},q_1,q_2}\left((\tx{LR}[\hat{F}], \hat{f});(\Pi, \hat{f})\right)\leq\\
&\leq \Delta^{\pm}_{q_{*},q_1,q_2}\left((\tx{LR}[\hat{F}], \hat{f});(\tx{LR}[\hat{\Pi}], \hat{f})\right)+\Delta^{\pm}_{q_{*},q_1,q_2}\left((\tx{LR}[\hat{\Pi}], \hat{f});(\Pi, \hat{f})\right) \leq\\
&\leq  \frac{4 q_{F} \left(2q_{F} +q_1+ q_2\right)}{2^n}+\Delta^{\pm}_{q_{*}}\left(\tx{LR}[\hat{\Pi}];\Pi\right) \le \\
&\leq \frac{4q_{*}(2q_{*}+q_1+q_2)}{2^n}+\frac{5q_{*}^2}{2^n}= \frac{q_{*}(13q_{*}+4q_1+4q_2)}{2^n}.
\end{align*}
\end{proof}

The same argument can be used to show a similar bound for the single permutation $2$-rounds BPEM cipher.
\begin{theorem}
\label{thm:indistinguish_1perm}
Consider the single permutation BPEM cipher $\tx{BPEM}[K_0, K_1, K_2;f,f]$ where the (secret) keys
$K_0, K_1, K_2$ are selected uniformly and independently at random.
Let $q_{*}$ be the maximum number of queries to the encryption/decryption oracle, and let $q$ be the maximum number  of queries to the public permutation $f$. Then,
$$
\Delta^{im\textit{-}sprp}_{\tx{BPEM}[K_0,K_1,K_2;f,f]}(q_{*},q)\leq
\frac{q_{*}(21q_{*}+8q)}{2^n}
$$
\end{theorem}
\begin{remark}
The difference in the bounds we received in Theorems \ref{thm:indistinguish} and \ref{thm:indistinguish_1perm}  are due only to the difference in the value of $\sigma$ in the application of Lemma \ref{lemma3}.
\end{remark}

We also comment that the same bounds hold in the single key variants. By \eqref{matrix1} we have
\begin{align*}
\tx{BPEM}[K,K,K;f_1,f_2]& = \tx{LR}[f_1^{\oplus K'_1}, f_1^{\oplus K'_2}, f_2^{\oplus K'_2}, f_2^{\oplus K'_3}],\\
\tx{BPEM}[K,K,K;f,f] &= \tx{LR}[f^{\oplus K'_1}, f^{\oplus K'_2}, f^{\oplus K'_2}, f^{\oplus K'_3}]
\end{align*}
where
\begin{center}
$\begin{pmatrix}
K'_{1} \\
K'_{2} \\
K'_{3}
\end{pmatrix}
= \begin{pmatrix}
1 & 0 \\
1 & 1 \\
0 & 1
\end{pmatrix}
\cdot
\begin{pmatrix}
K_{R} \\
K_{L}
\end{pmatrix}
$.
\end{center}

For both constructions, the ``new" keys $K'_{1},K'_{2},K'_{3}$ are no longer independent, so we need to generalize lemma \ref{lemma3} as stated below.

\begin{lemma}\label{lemma4}
Let $\Pi_1, \ldots, \Pi_r, \bar{\Pi}_1, \ldots, \bar{\Pi}_t$ be independent random permutations of $\s^n$ and $K_1, \ldots K_s$ be chosen uniformly and independently from $\s^n$.
We write $\hat{\Pi}$ to denote $(\Pi_1, \ldots, \Pi_r)$.
Let $(\alpha_1, \ldots, \alpha_t)$ be a sequence of elements from $\{1,\ldots,r\}$.
Let $M$ be a binary matrix of size $t \times s$, with no zero rows, satisfying the following condition: for every $1\leq i_1<i_2\leq t$ such that $\alpha_{i_1}=\alpha_{i_2}$, the $i_1^{\rm{th}}$ and $i_2^{\rm{th}}$ rows of $M$ are distinct.
Let $K'_i:=\sum_{j=1}^s M_{ij}K_j$, for every $1\leq i\leq t$ .

Then, for any $\theta = (q_1, \ldots, q_t, q'_1, \ldots, q'_r)$ (specifying the maximum number of queries) we have
\[\Delta^{\pm}_{\theta}\left((\bar{\Pi}_1, \ldots, \bar{\Pi}_t, \hat{\Pi});(\Pi_{\alpha_1}^{\oplus K'_{\beta_1}}, \ldots, \Pi_{\alpha_t}^{\oplus K'_{\beta_t}}, \hat{\Pi})\right) \leq  \frac{\sigma}{2^n}\] where
$\sigma :=2\sum_{\alpha=1}^r\left(\binom{\sigma_{\alpha}}{2}+\sigma_{\alpha}q'_{\alpha}\right)$
and $\sigma$ is as defined in Lemma \ref{lemma3}.
\end{lemma}

We skip the proof of this lemma as it is similar to that of Lemma \ref{lemma3}.
Similarly to the proof of Theorem \ref{thm:indistinguish} (while using Lemma \ref{lemma4} instead of Lemma \ref{lemma3}), we can obtain the following bound.
\begin{theorem}
\label{thm:indistinguish_1key}
Consider the single key BPEM cipher $\tx{BPEM}[K, K, K;f_1,f_2]$ where the (secret) key
$K$ is selected uniformly at random.
Let $q_{*}$ be the maximum number of queries to the encryption/decryption oracle, and let $q_1, q_2$ be the maximum numbers of queries to the public permutations $f_1$ and $f_2$, respectively. Then, $$\Delta^{im\textit{-}sprp}_{\tx{BPEM}[K,K,K;f_1,f_2]}(q_{*},q_1,q_2)\leq \frac{q_{*}(13q_{*}+4q_1+4q_2)}{2^n}.$$
\end{theorem}

Finally, similarly to the proof of Theorem \ref{thm:indistinguish_1perm} (while using Lemma \ref{lemma4} instead of Lemma \ref{lemma3}, and using Theorem \ref{Nandi} instead of Theorem \ref{Piret}), we obtain the following bound.
\begin{theorem}
\label{thm:indistinguish_1key_1perm}
Consider the single key single permutation BPEM cipher $\tx{BPEM}[K, K, K;f,f]$ where the (secret) key
$K$ is selected uniformly at random.
Let $q_{*}$ be the maximum number of queries to the encryption/decryption oracle, and let $q$ be the maximum number  of queries to the public permutation $f$. Then,
$$\Delta^{im\textit{-}sprp}_{\tx{BPEM}[K,K,K;f,f]}(q_{*},q)\leq \frac{q_{*}(16q_{*}+8q)}{2^n}+\frac{17q_{*}^2}{2^n-1} + \frac{q_{*}^2}{2^{2n}}.$$
\end{theorem}

\section{A distinguishing attack on BPEM}\label{sec:attack}

In this section we describe a distinguishing attack on BPEM that uses $O(2^{n/2})$ queries. This is the same attack as the one described in 
\cite[Section 3.2]{TP} for the $4$-rounds Luby-Rackoff with internal permutations, not at all surprising, since we showed (in Section \ref{sec:BPEM_LR}) that BPEM can be viewed as a $4$-rounds Luby-Rackoff with internal (keyed) permutations. Nevertheless, for the sake of completeness, we describe and analyze the attack in this BPEM terminology.
We will use the following technical lemma.
\begin{lemma}\label{attack_lemma}
If $x,y,\rho\in\s^n$ such that
\begin{gather}\label{attack_lemma_given}
x\oplus\left(\tx{BPEM}[K_0,K_1, K_2; f_1, f_2](x*\rho)\right)_L=\nonumber\\
=y\oplus\left(\tx{BPEM}[K_0,K_1, K_2; f_1, f_2](y*\rho)\right)_L
\end{gather}
then $x=y$.
\end{lemma}
\begin{proof}
Denote
\begin{gather*}
\check{x}:=\tx{LR}^{\,2}[f_1]\left((x*\rho)\oplus K_0\right)\oplus K_1,\\
\check{y}:=\tx{LR}^{\,2}[f_1]\left((y*\rho)\oplus K_0\right)\oplus K_1.
\end{gather*}
By \eqref{eq:def_of_BPEM} and \eqref{eq:def_of_EM} we have that
\begin{gather*}
\tx{BPEM}[K_0,K_1, K_2; f_1, f_2](x*\rho)=
\tx{EM}^{\tx{LR}^{\,2}[f_1],\tx{LR}^{\,2}[f_2]}_{K_0, K_1, K_2}(x*\rho)=\\
=\tx{LR}^{\,2}[f_2]\left(\tx{LR}^{\,2}[f_1]((x*\rho)\oplus K_0)\oplus K_1\right)\oplus K_2=\tx{LR}^{\,2}[f_2]\left(\check{x}\right)\oplus K_2,
\end{gather*}
hence, by \eqref{feistel2},
\begin{gather*}
\left(\tx{BPEM}[K_0,K_1, K_2; f_1, f_2](x*\rho)\right)_L
=\left(\tx{LR}^{\,2}[f_2]\left(\check{x}\right)\oplus K_2\right)_L=\\
=\check{x}_L\oplus f_2\left(\check{x}_R\right)\oplus(K_2)_L=x\oplus(K_0)_L\oplus (K_1)_L\oplus f_1\left(\rho\oplus(K_0)_R\right)\oplus f_2\left(\check{x}_R\right)\oplus(K_2)_L.
\end{gather*}
Similarly
\begin{gather*}\left(\tx{BPEM}[K_0,K_1, K_2; f_1, f_2](y*\rho)\right)_L=\\
=y\oplus(K_0)_L\oplus (K_1)_L\oplus f_1\left(\rho\oplus(K_0)_R\right)\oplus f_2\left(\check{y}_R\right)\oplus(K_2)_L.
\end{gather*}
Therefore we get from \eqref{attack_lemma_given} that $ f_2\left(\check{x}_R\right)= f_2\left(\check{y}_R\right)$, hence, since $f_2$ is injective, $\check{x}_R=\check{y}_R$. Therefore, using \eqref{feistel2} again,
\begin{gather*}
\rho\oplus(K_0)_R\oplus f_1\left(x\oplus(K_0)_L\oplus f_1(\rho\oplus(K_0)_R)\right)\oplus (K_1)_R=\\
=\rho\oplus(K_0)_R\oplus f_1\left(y\oplus(K_0)_L\oplus f_1(\rho\oplus(K_0)_R)\right)\oplus (K_1)_R,
\end{gather*}
hence
\begin{gather*}
f_1\left(x\oplus(K_0)_L\oplus f_1(\rho\oplus(K_0)_R)\right)=f_1\left(y\oplus(K_0)_L\oplus f_1(\rho\oplus(K_0)_R)\right).
\end{gather*}
Since $f_1$ is injective we get that
$$x\oplus(K_0)_L\oplus f_1(\rho\oplus(K_0)_R)=y\oplus(K_0)_L\oplus f_1(\rho\oplus(K_0)_R),$$
hence $x=y$.
\end{proof}
\begin{proposition}\label{propos:attack}
Consider the BPEM cipher $\tx{BPEM}[K_0, K_1, K_2;f_1,f_2]$ with arbitrary (secret) keys
$K_0, K_1, K_2$.
Let $q_{*}$ be the maximum number of queries to the encryption oracle. Then,
$$
\Delta^{prp}_{\tx{BPEM}}(q_{*})\geq
1-e^{-\frac{q_{*}(q_{*}-1)}{2(2^n+1)}}.
$$
\end{proposition}
\begin{remark}
Note that Proposition \ref{propos:attack} implies that the adversary advantage becomes non-negligible for $q_{*}=\Omega(2^{n/2})$.
\end{remark}
\begin{proof}
 Fix an $n$-bit string $\rho$ and $q_{*}$ distinct $n$-bit strings $\omega_1, \omega_2,\ldots,\omega_{q_{*}}$. We query the encryption oracle for the plaintexts $\omega_1*\rho,\omega_2*\rho,\ldots,\omega_{q_{*}}*\rho$, and let $\sigma_1,\sigma_2,\ldots,\sigma_{q_{*}}$ be the corresponding ciphertexts. We now search for collisions between the ${q_{*}}$ $n$-bit strings $\omega_1\oplus (\sigma_1)_L,\omega_2\oplus (\sigma_2)_L,\ldots,\omega_k\oplus (\sigma_{q_{*}})_L$.
By Lemma \ref{attack_lemma} there will be no collision if the oracle encrypts using $\tx{BPEM}[K_0,K_1, K_2; f_1, f_2]$. By contrast, if the oracle encrypts by applying a randomly chosen permutation of $\s^{2n}$ then the probability there is no collision is at most
\begin{gather*}
\prod_{k=1}^{{q_{*}}-1} \left(1-\frac{k(2^n-1)}{2^{2n}-k}\right)\leq \prod_{k=1}^{{q_{*}}-1} \left(1-\frac{k}{2^n+1}\right)\leq\prod_{k=1}^{{q_{*}}-1} e^{-\frac{k}{2^n+1}}=e^{-\frac{{q_{*}}({q_{*}}-1)}{2(2^n+1)}}.
\end{gather*}
\end{proof}

\section{A practical constructions of a $256$-bit cipher}
\label{sec:practical}

In this section, we demonstrate a practical construction of a $256$-bit block cipher based on the $2$-rounds BPEM cipher, where the underlying permutation is AES.

\begin{definition}[$EM256AES$: a $256$-bit block cipher]
Let $\ell_1$ and $\ell_2$ be two $128$-bit keys and let $K_0, K_1, K_2$ be three $256$-bit secret keys (assume $\ell_1,\ell_2,K_0,K_1,K_2$ are selected uniformly and independently at random). Let the permutations $f_1$ and $f_2$ be the AES encryption using $\ell_1$ and $\ell_2$ as the AES key, respectively. \\
The 256-bit block cipher $EM256AES$ is defined as the associated instantiation of the $2$-rounds BPEM cipher $\tx{BPEM}[K_0,K_1, K_2;f_1, f_2] $. \\
Usage of $EM256AES$:

\begin{itemize}
\item [$\bullet$]
$\ell_1$ and $\ell_2$ are determined during the setup phase, and can be made public (e.g., sent from the sender to the receiver as an IV).
\item [$\bullet$]
$K_0, K_1, K_2$ are selected per encryption session.
\end{itemize}
The single key EM256AES is the special case where a single value $K \in \{0, 1 \}^{256}$ and a single value $\ell \in \{0, 1 \}^{128}$ are selected uniformly and independently at random, and the EM256AES cipher uses $K_0=K_1=K_2 = K$ and $\ell_1=\ell_2 = \ell$.
\end{definition}

Hereafter, we use the single key EM256AES. To establish security properties for $EM256AES$, we make the standard assumption about AES with a secret key that is selected (uniformly at random): an adversary has negligible advantage in distinguishing AES from a random permutation of $\{0, 1\}^{128}$ even after seeing a (very) large number of plaintext-ciphertext pairs (i.e., the assumption is that AES satisfies its design goals (\cite{NIST_request}, Section 4).
This assumption is certainly reasonable if the number of blocks that are encrypted with the same keys is limited to be much smaller than $2^{64}$.
Therefore, in our context, we can consider assigning the randomly selected key $\ell$ at setup time, as an approximation for a random selection of the permutation $f_1 = f_2$.
Under this assumption, we can rely on the result of Theorem \ref{thm:indistinguish_1key_1perm} for the security of $EM256AES$.

\subsubsection{$EM256AES$ efficiency:}
An encryption session between two parties requires exchanging a $256$-bit secret key and transmitting a $128$-bit IV ($=\ell$). One key (and IV) can be used for $N$ blocks as long as we keep $N \ll 2^{64}$. \\
Computing one ($256$-bit) ciphertext involves $4$ AES computations (with the $IV$ as the AES key) plus a few much cheaper XOR operations. Let us assume that the encryption is executed on a platform that has the capability of computing AES at some level of performance. If the $EM256AES$ encryption (decryption) is done in a serial mode, we can estimate the encryption rate to be roughly half the rate of AES (serial) computation on that platform ($4$ AES operations per one $256$-bit block). Similarly, if the $EM256AES$ encryption is done in a parallelized mode, we can estimate the throughput to be roughly half the throughput of AES.

\subsubsection{$EM256AES$ performance:}
To test the actual performance of $EM256AES$, and validate our predictions, we coded an optimized implementation of \\
$EM256AES$. Its performance is reported here. \\
The performance was measured on an Intel� Core� i7-4700MQ (microarchitecture Codename Haswell) where the enhancements (Intel� Turbo Boost Technology, Intel� Hyper-Threading Technology, and Enhanced Intel Speedstep� Technology) were disabled. The code used the AES instructions (AES-NI) that are available on such modern processors. \\
On this platform, we point out the following baseline: the performance of AES ($128$-bit key) in a parallelized mode (CTR) is $0.63$ C/B, and in a serial mode (CBC) it is $4.44$ cycles per byte (C/B hereafter). \\
The measured performance of our $EM256AES$ implementation was $1.44$ C/B for the parallel mode, and $8.92$ C/B for the serial mode. The measured performance clearly matches the predictions. \\
It is also interesting to compare the performance of $EM256AES$ to another $256$-bit cipher. To this end, we prepared an implementation of Rijndael256 cipher \cite{Rijndael256}
\footnote[2]{AES is based on the Rijndael block cipher. While AES standardizes only a $128$ block size, the Rijndael definitions support both $128$-bit and $256$-bit blocks}. For details on how to code Rijndael256 with AES-NI, see \cite{Gueron_WP}). Rijndael256 (in ECB mode) turned out to be much slower than $EM256AES$, performing at $3.85$ C/B.

\section{Discussion}\label{sec:discussion}

In this work, we showed how to construct a large family of balanced permutations, and analyzed the resulting new variation, BPEM, of the EM cipher. 

The resulting $2n$-bit block cipher is obtained by using a permutation of $\s^n$ as a primitive.
The computational cost of encrypting (decrypting) one $2n$-bit block is $4$ evaluations of a permutation of $\s^n$ (plus a relatively small overhead).
Note that this makes BPEM readily useful in practice, for example to define a $256$-bit cipher, because ``good'' permutations of $\{0, 1\}^{128}$ are available. We demonstrated the specific cipher $EM256AES$, which is based on AES, and showed that its throughput is (only) half the throughput of AES (and $2.5$ times faster than Rijndael256).

A variation on the way by which BPEM can be used, would make it a tweakable $2n$-bit block cipher. Here, the public IV (=$\ell$) can be associated with each encrypted block as an identifier, to be viewed as the tweak. The implementation would switch this tweak for each block. To randomize the keys for the (public) permutations, an additional encryption (using some secret key) is necessary.

The expression of the advantage in Theorem \ref{thm:indistinguish} behaves linearly with the number of queries to the public permutations, and quadratically with the number of queries to the encryption/decryption oracle.
This reflects the intuition that the essential limitations on the number of adversary queries should be on the encryption/decryption invocations, while weaker (or perhaps no) limitations should be imposed on the number of queries to the public permutations. 
It also suggests the following protocol, where the secret keys are changed more frequently than the random permutations.
{\it 
Choose the public permutations for a period of, say, $\frac{1}{1000}2^{2n/3}$ blocks, divided into $2^{n/3}$ sessions of $\frac{1}{1000}2^{n/3}$ blocks. Change the secret keys every session.}
This way, the relevant information on the block cipher, from a specific choice of keys, is limited to a session, while the adversary can accumulate relevant information from replies to the public permutations across sessions. Therefore, $q_{*}$ is limited to $\frac{1}{1000}2^{n/3}$, while $q_{*}+q_1+q_2$ is limited to $\frac{1}{1000}2^{2n/3}$. Theorem \ref{thm:indistinguish} guarantees that this usage is secure.

\end{document}